\documentclass[11pt]{article}
\usepackage{authblk}
\usepackage{graphicx}
\usepackage[utf8]{inputenc} 
\usepackage{color}
\usepackage{amsmath}
\usepackage{amsthm}
\usepackage{amssymb}
\usepackage{fullpage}
\usepackage{multirow}
\usepackage{enumerate}

\newtheorem{Th}{Theorem}

\newtheorem{Lem}[Th]{Lemma}

\theoremstyle{remark}

\newtheorem{R}{Remark}
\newtheorem{Cor}[Th]{Corollary}
\theoremstyle{definition}

\newtheorem*{Ex}{Example}

\def\tr{\textnormal{tr}}

\begin{document}

\title{Pure states of maximum uncertainty with respect to a given POVM}
\author[1]{Anna Szymusiak \thanks{anna.szymusiak@uj.edu.pl}}
\affil[1]{Institute of Mathematics, Jagiellonian University, \L ojasiewicza 6, 30-348 Krak\'ow, Poland}
\date{}

\maketitle
\begin{abstract}
One of the differences between classical and quantum world is that in the former we can always perform a measurement that gives certain outcomes for all pure states, while such a~situation is not possible in the latter. The degree of randomness of the distribution of the measurement outcomes can be quantified by the Shannon entropy. While it is well known that this entropy, as a function of quantum states, needs to be minimized by some pure states, we would like to address the question how `badly' can we end by choosing initially any pure state, i.e., which pure states produce the maximal amount of uncertainty under given measurement. We find these maximizers for all highly symmetric POVMs in dimension 2, and for all SIC-POVMs in any dimension. 
\end{abstract}

One of the differences between classical and quantum world is that classically we can always perform a measurement that gives certain outcomes for all pure states, while such a situation is not possible in a quantum setup. In this paper we are interested in the question to what extend can the outcomes be uncertain for fixed quantum measurement, if it was performed on the pure state. It is natural to think of measuring uncertainty by quantifying the randomness of the probability distribution of the measurement outcomes, e.g. by calculating its Shannon entropy. In this context, the maximizers of this entropy can be described as the states of maximal uncertainty and the minimizers as these of minimal uncertainty.

The problem of minimization of the Shannon entropy of a quantum measurement, mathematical description of which is given by the positive operator-valued measure (POVM), has been already widely considered, in particular in the context of the informational power of POVM, and it has been solved for all highly symmetric POVMs in dimension two, i.e., seven sporadic measurements \cite{SloSzy16}, including the `tetrahedral' SIC-POVM \cite{DAretal11,Oreetal11} and one infinite series \cite{SloSzy16,DAretal11}, all SIC-POVMs in dimension three \cite{Szy14}, the POVM consisting of four MUBs in dimension three \cite{DAr14} and the Hoggar SIC-POVM in dimension eight \cite{SzySlo16}. In this paper we focus on the opposite problem, i.e., the one regarding maximization of the Shannon entropy of POVM over the pure states. We solve it for all highly symmetric POVMs in dimension 2, and, quite surprisingly, for all SIC-POVMs (in any dimension, assuming that such objects exist). While the  latter result turns out to be easily solvable, we have not been able to find it published anywhere and we are convinced it deserves to be noticed as it seems to be of some importance for at least two reasons. Firstly, SIC-POVMs exhibit properties that make them applicable in quantum state tomography \cite{Sco06,ZhuEng11,JiaPhd,Benetal15}, quantum cryptography \cite{Ren05}, quantum communication \cite{Oreetal11} or entanglement detection \cite{Cheetal15}, they are also essential component of quantum Bayesianism \cite{Cavetal02}. Secondly, there are only a few universal results true for all SIC-POVMs, in contrast to many open problems, starting with the crucial question whether at all they exist in all dimensions.

The maximization of the Shannon entropy of POVM is also closely related to the entropic \emph{certainty} relations \cite{San95,Pucetal15}. In fact, such a relation can be derived from any upper bound on the entropy of POVM  whenever this POVM can be decomposed (up to a scalar factor) into $m$ POVMs consisting of the same number of elements.

\section{Entropy of POVM}

Let us consider a finite-dimensional quantum system with the associated Hilbert space $\mathbb C^d$ and denote by $\mathcal S(\mathbb C^d)$ the set of all quantum states (density matrices), and by $\mathcal P(\mathbb C^d)$ the set of pure states. The mathematical description of a general quantum measurement is given by a \emph{positive operator valued measure} (POVM). In a finite case, it is enough to say that a POVM is a set $\Pi=\{\Pi_j\}_{j=1}^k$ of nonzero positive semidefinite operators on $\mathbb C^d$ (called \emph{effects}) satisfying the identity decomposition $\sum_{j=1}^k\Pi_j=\mathbb I_d$. While the indices $j=1,\ldots,k$ label the outcomes, the probability of obtaining the $j$-th outcome is given by $p_j(\rho,\Pi):=\tr(\rho\Pi_j)$, providing the system before measurement was in the state $\rho\in\mathcal S(\mathbb C^d)$. 

We say that a POVM is \emph{rank-$1$} and \emph{normalized} if its elements are rescaled one-dimensional projectors with $\tr(\Pi_j)=d/k$ for $j=1,\ldots,k$. With any such measurement we can associate the set of pure states $\rho_j:=(k/d)\Pi_j$ ($j=1,\ldots,k$), which allows us to think of a representation of the POVM on the (generalized) Bloch sphere. A special case of such POVMs are \emph{projective-valued measures} (PVMs), where $k=d$ and so $\Pi$ consists of rank-1 orthogonal projections onto an orthonormal basis. We call a POVM \emph{informationally complete} (and denote IC-POVM) if its statistics determine uniquely the initial state, i.e., the conditions $\tr(\rho\Pi_j)=\tr(\sigma\Pi_j)$ for $j=1,\ldots,k$ imply $\rho=\sigma$ for any $\rho,\sigma\in\mathcal S(\mathbb C^d)$. A rank-1 normalized POVM consisting of $d^2$ effects is called \emph{symmetric informationally complete} POVM (SIC-POVM) if $\tr(\Pi_i\Pi_j)=1/(d^2(d+1))$ for all \mbox{$i\neq j$}. Note that SIC-POVMs are indeed informationally complete. We distinguish also a class of POVMs that exhibit an especially high level of symmetry and call them  \emph{highly symmetric} POVMs (HS-POVMs) \cite{SloSzy16}. In two-dimensional case all HS-POVMs can be nicely characterized by the configurations of the associated pure states on the Bloch sphere. Namely, there are five HS-POVMs represented by the vertices of Platonic solids (tetrahedron, cube, octahedron, icosahedron and dodecahedron) and two represented by quasiregular Archimedean solids: cuboctahedron and icosidodecahedron, all of which are informationally complete. There is also an infinite series of non-informationally complete HS-POVMs, represented by the vertices of regular equatorial polygons (including digon). For more information see \cite{SloSzy16}. 

In order to quantify the randomness of the obtained probability distribution (and so the uncertainty of the measurement outcomes) we define the \emph{entropy of the measurement} $\Pi$, namely the function $H(\cdot,\Pi):\mathcal S(\mathbb C^d)\to\mathbb R$ which evaluates the Shannon entropy of the probability distribution of measurement outcomes for the initial state $\rho$, i.e.,
$$H(\rho,\Pi):=\sum_{j=1}^k\eta(p_j(\rho,\Pi)),$$
where $\eta(x)=-x\ln x$ for $x>0$ and $\eta(0):=0$. 

The entropy of the measurement $H(\cdot,\Pi)$ is continuous and concave, and so  its minima occur at extreme points of $\mathcal S(\mathbb C^d)$, i.e., the pure states. Obviously, this quantity is upper bounded by $\ln k$, the entropy of the uniform distribution. This upper bound is achieved in the maximally mixed state $\rho_*:=(1/d)\mathbb I_d$ whenever $\Pi$ consists of operators of equal trace, i.e., $\tr(\Pi_j)=d/k$ for $j=1,\ldots,k$, since then $\tr(\rho_*\Pi_j)=(1/d)\tr(\Pi_j)=(1/d)(d/k)=1/k$, and so the measurement outcomes are uniformly distributed. In particular, it is satisfied by rank-1 normalized POVMs.  However, if we consider  the entropy of the measurement restricted to the pure states, the question of which pure states maximize the entropy of the measurement and how large it can be is not so trivial. What is the meaning of this question? Let us observe that, since the entropy is minimized on the set of pure states, one can ask how `badly' we can end up by choosing initially any pure state.

Note also that instead of entropy we can analyse the \emph{relative entropy  of the measurement} (with respect to the uniform distribution), i.e., $$\widetilde{H}(\cdot,\Pi):=\ln k -H(\cdot,\Pi).$$ Obviously, the maximization of the entropy $H$ over pure states is equivalent to the minimization of the relative entropy $\widetilde H$ over pure states, the latter being lower bounded by 0. It seems it makes no difference which of these two functions we take into consideration. However, the relative entropy can be a more convenient tool since its average value over pure states depends only on the dimension $d$ (see, e.g.\ \cite{SloSzy16}) and so it enables us  to compare, in some sense, the indeterminacy of different  measurements on the same quantum system. 

An obvious question that arises here is whether we can compute exactly the maximum entropy of the measurement, i.e., the minimum relative entropy $\widetilde H$, for at least some classes of POVMs and determine the maximizers (minimizers of $\widetilde{H}$). 
Let us start with the following simple observations.

\begin{R}\label{fact}\hfill
	\begin{enumerate}[I.]
		
		\item If $\Pi$ is a PVM, then $\min_{\rho\in\mathcal P(\mathbb C^d)} \widetilde H(\rho,\Pi)=0$, ($\max_{\rho\in\mathcal P(\mathbb C^d)}  H(\rho,\Pi)=\ln d$)  and the set of minimizers (maximizers) is a $(d-1)$-torus.
		\item \label{fact2} If a rank-1 normalized POVM $\Pi=\{\Pi_j\}_{j=1}^k$ is informationally complete, then 
		\begin{equation}
		\min_{\rho\in\mathcal P(\mathbb C^d)} \widetilde H(\rho,\Pi)>0,\quad (\max_{\rho\in\mathcal P(\mathbb C^d)}  H(\rho,\Pi)<\ln k).
		\end{equation}
		Moreover, if $d=2$, the converse is also true.
	\end{enumerate}
\end{R}

\begin{proof}\hfill
	
	\noindent I. If $\Pi$ is a PVM, i.e., $\Pi:=\{|e_j\rangle\langle e_j|\}_{j=1}^d$, where $\{|e_j\rangle\}_{j=1}^d$ is an orthonormal basis in $\mathbb C^d$, then the measurement outcomes are uniformly distributed if and only if the initial state $|\psi\rangle\langle\psi|$ is of the form $|\psi\rangle=(1/\sqrt d)\sum_{j=1}^de^{i\theta_j}|e_j\rangle$, where $\theta_j\in\nolinebreak\mathbb R$. Thus, the minimizers of $\tilde{H}$ form a $(d-1)$-torus.
	
	\noindent II. If $\widetilde H(\rho,\Pi)=0$, then $H(\rho,\Pi)=\ln k$, and so the probability distribution of the measurement outcomes is uniform. By the informational completeness of $\Pi$ we get  $\rho=\rho_*$, and so $\min_{\rho\in\mathcal P(\mathbb C^d)} \widetilde H(\rho,\Pi)>0$.
	
	Now let $d=2$. To see the converse, let us assume that $\Pi$ is not informationally complete. Notice that the informational completeness of a rank-1 normalized POVM $\Pi=\{\Pi_j\}_{j=1}^k$ is equivalent to the condition that the operators $\Pi_j-(1/k)\mathbb I_d$ span the real (($d^2-1$)-dimensional) space of traceless selfadjoint operators, see \cite[Prop.\ 3.51]{HeiZim}.  Then, by the fact that $\mathcal P(\mathbb C^2)$ is a sphere centered at $\rho_*$ in the affine subspace of trace-one selfadjoint operators (the Bloch sphere), there exists  $\sigma\in \mathcal P(\mathbb C^2)$ such that $\sigma-(1/d)\mathbb I_d $ is orthogonal to $\Pi_j-(1/k)\mathbb I_d$ for  $j=1,\ldots,k$, i.e., $\tr((\sigma-(1/d)\mathbb I_d)(\Pi_j-(1/k)\mathbb I_d))=0$ for all $j$. In consequence, $\tr(\sigma\Pi_j)=1/k$ for all $j$ and so $\widetilde H(\sigma,\Pi)=0$.
\end{proof}

\begin{Ex}
	In order to see that  the converse of (\ref{fact2}) need not necessarily be true in higher dimensions one can consider some  \emph{strongly unextendible} sets of mutually unbiased bases (MUBs), i.e., the ones for which there does not exist even a single vector unbiased with respect to these bases. In consequence, such a set generates a measurement which is not informationally complete, and there exist no pure state providing a uniform distribution of the measurement outcomes. For a specific example, see e.g. \cite[Thm 2]{Gra04}.
\end{Ex}

The following corollary is a straightforward consequence of Remark \ref{fact}.
\begin{Cor}
	For any two-dimensional HS-POVM $\Pi$ represented by a regular equatorial $n$-gon ($n>2$) on the Bloch sphere we have $\min_{\rho\in\mathcal P(\mathbb C^d)}\widetilde{H}(\rho,\Pi)=0$, ($\max_{\rho\in\mathcal P(\mathbb C^d)} H(\rho,\Pi)=\ln n$ ) and the minimum (maximum) is achieved at the poles.
\end{Cor}

\section{States of maximum uncertainty with respect to HS-POVMs in dimension 2 and SIC-POVMs}

Finding extremal points of entropy-like functions is in general not an easy task, since the standard procedures, like the Lagrange multipliers method, fail. The main idea of the proofs presented below is to replace the entropy function $\eta$ with a polynomial in such a way that the obtained modified function agrees with $H(\cdot,\Pi)$ exactly at the global maximizers and is greater otherwise. To achieve this aim we use the method based on the Hermite interpolation described in detail in \cite{SloSzy16}. We give a brief background on the Hermite interpolation in  Appendix \ref{hermite}.

All the results presented below refer to rank-1 normalized POVMs and so we find it more convenient to consider the entropy of a measurement as a function of its Bloch vectors. We use the normalized version of the Bloch representation, i.e., we assume that the generalized Bloch set $B(d)$ is a subset of $(d^2-2)$-dimensional real \emph{unit} sphere (in particular, $B(2)$ is just the unit sphere $S^2$). In this setup the relation between the Hilbert-Schmidt inner product of states $\rho$ and $\sigma$ and the standard inner product of their Bloch images $u$ and $v$ is as follows: $\langle\langle\rho,\sigma\rangle\rangle_{HS}=\tr(\rho\sigma)=((d-1)u\cdot v+1)/d$. Let us denote by $v_1,\ldots,v_k$ the Bloch representations of states $\rho_1,\ldots,\rho_k$ associated with the POVM $\Pi=\{\Pi_j\}_{j=1}^k$ and by $u$ the Bloch representation of some state $\rho$. Then $p_j(\rho,\Pi)=((d-1)u\cdot v_j)+1)/k$ and
\begin{equation}\label{entropyBloch}
H_{\Pi}(u):=H(\rho,\Pi)=\sum_{j=1}^k\eta\left(\frac{(d-1)u\cdot v_j+1}{k}\right)=\ln\frac{k}{d}+\frac{d}{k}\sum_{j=1}^k h(u\cdot v_j),
\end{equation}
where $h:[-1/(d-1),1]\to\mathbb R^+$ is defined by $h(t):=\eta(((d-1)t+1)/d)$.

Firstly, we prove a simple result that the minimum relative entropy (and so maximum entropy) of a SIC-POVM is always attained at the states constituting this SIC-POVM: 

\begin{Th}\label{maxsic}
	Let  $\Pi=\{\Pi_j\}_{j=1}^{d^2}$ be a SIC-POVM in dimension~$d$. Then the states $\rho_j:=d\Pi_j$ for $j=1,\ldots,d^2$ are the only minimizers of the relative entropy of $\Pi$ restricted to the pure states and
	\begin{equation}\label{sicbound}
	\min_{\rho\in\mathcal P(\mathbb C^d)}\widetilde H(\rho,\Pi)=\widetilde H(\rho_j,\Pi)=\ln d-\frac{d-1}{d}\ln(d+1),
	\end{equation}
	for $j=1,\ldots,d^2$. Moreover,  $\min_{\rho\in\mathcal P(\mathbb C^d)}\widetilde H(\rho,\Pi)\xrightarrow{d\to\infty} 0$.
\end{Th}

\begin{proof}
	From the definition of SIC-POVM we obtain $p_i(\rho_j,\Pi)=\tr((1/d)\rho_j\rho_i)= 1/(d(d+1))$ for $i\neq j$ and $p_j(\rho_j,\Pi)=1/d$. Thus, we get the second equality in (\ref{sicbound}). To see that it is indeed the minimum value of the relative entropy on $\mathcal P(\mathbb C^d)$, let us use the Hermite interpolation method. We consider the entropy of $\Pi$ redefined in (\ref{entropyBloch}) 
	to be a function of Bloch vectors, putting $k=d^2$. Note that $v_j\cdot v_i=-1/(d^2-1)$ and $v_j\cdot v_j=1$. 
	According to Lemma \ref{fromabove} (see Appendix \ref{hermite}), the interpolating Hermite polynomial $p$ such that $p(1)=h(1)$, $p(-1/(d^2-1))= h(-1/(d^2-1))$ and $p'(-1/(d^2-1))=h'(-1/(d^2-1))$   interpolates $h$ from above.
	Thus, the polynomial function given by $P(u):=\ln d+\frac{1}{d}\sum_{j=1}^{d^2}p(u\cdot v_j)$ for $u\in B(d)\subset\mathbb R^{d^2-1}$ is of degree less than 3. Since every SIC-POVM is a projective 2-design \cite{Renetal04}, $\sum_{j=1}^{d^2}(u\cdot v_j)^t=const$ for $t=1,2$, and so $P$ is necessarily constant on the whole Bloch set. Using the fact that $P(u)\geq H_{\Pi}(u)$ for $u\in B(d)$, and $P(v_j)=H_\Pi(v_j)$ for $j=1,\ldots,d^2$, we conclude that the entropy attains its maximum (and so the relative entropy its minimum) on the set of pure states  at $\rho_j$ ($j=1,\ldots,d^2$).
	
	To see that there are no other maximizers, let us observe that 
	if a pure state $\rho$ is also a global maximizer of the entropy, then $\{u\cdot v_j|j=1,\ldots,d^2\}$ has to be contained in the set of interpolating points $T:=\{-1/(d^2-1),1\}$, where $u$ is a normalized Bloch vector corresponding to $\rho$. Using the fact that $\sum_{j=1}^{d^2}u\cdot v_j=0$, we get $\{u\cdot v_j|j=1,\ldots,d^2\}= T$, and so $u\in\{v_1,\ldots,v_{d^2}\}$. Thus, the uniqueness is proven. The limit as $d\to\infty$ follows from direct calculation.
\end{proof}

\begin{R}\label{genentr} 
	Similar statement, concerning the arrangement of the maximizers of the entropy holds true if we replace the Shannon entropy with the Havrda-Charv\'at-Tsallis $\alpha$-entropy $H_\alpha$ or the R\'enyi $\alpha$-entropy $R_\alpha$ for $\alpha\in (0,2)$ (note that for $\alpha=1$ we get the Shannon entropy). The reason is that the Hermite polynomial of the function $\theta_\alpha(x):=(x-x^\alpha)/(\alpha-1)$ for $\alpha\in(0,2)\setminus\{1\}$ used in the definition of $H_\alpha$ is also interpolating from above (see Lemma \ref{fromabove}) and since its degree enforces the resulting polynomial on the Bloch set to be constant, that completes the proof. Furthermore, nothing changes if we take a strictly increasing function of $H_\alpha$, which is the case if we consider $R_\alpha$ for $\alpha\in (0,2)$.
\end{R}

\begin{R}
	It is possible to complete the proof of Theorem \ref{maxsic} in another way, using the fact that the sum of squared probabilities of the measurement outcomes is the same for each initial pure state and equal to $2/(d(d+1))$. Theorem 2.5. from \cite{HarTop01} provides us with both minimizers and maximizers of the Shannon entropy of probability distribution under assumption that the \emph{index of coincidence}, i.e., the sum of squared probabilities is fixed. One might think that this is not particularly useful, since the set of  probability distributions \emph{allowed} by the SIC-POVM for initial pure states form just a $(2d-2)$-dimensional subset of a $(d^2-2)$-dimensional intersection of a $(d^2-1)$-sphere and the simplex $\Delta_{d^2}$. As $2d<d^2$ for $d>2$, it seems highly unlikely that the extremal points described in \cite{HarTop01}  belong to it. Thus it is quite surprising to find out that the probability distributions generated by the states indicated in Theorem \ref{maxsic} coincide with the ones optimal for general case.
\end{R}

Our next theorem provides the global maximizers (minimizers) of the entropy (relative entropy) of informationally complete HS-POVMs in dimension two. All of them have been already indicated as local minimizers in \cite[Proposition 8]{SloSzy16}. However, not all local minimizers found there turn out to be the global ones, as we see in the cuboctahedral and icosidodecahedral case.

\begin{Th}\label{HSThm}
	Let $\Pi$ be an informationally complete HS-POVM in dimension two, but not a SIC-POVM.  Then the entropy (resp.\ relative entropy) of $\Pi$  restricted to the set of pure states attains its maximum (resp.\ minimum) value exactly in the states whose Bloch vectors correspond to 
	\begin{enumerate}[1\emph{)}]
		\item the vertices of the dual polyhedron, if $\Pi$ is represented by a platonic solid,\label{platonic}
		\item the vertices of the octahedron, if $\Pi$ is represented by cuboctahedron,
		\item the vertices of the icosahedron, if $\Pi$ is represented by icosidodecahedron.
	\end{enumerate}
\end{Th}

Before we proceed to the proof, let us recall some basic facts concerning the rings of $G$-invariant polynomials. 
Since  only HS-POVMs in dimension 2 are considered here, we focus on the rings of $G$-invariant polynomials for the following groups: $T_d,O_h$ and $I_h$. For more general overview of the topic see, e.g., \cite{DerKem02,Jaretal84}. 
All the groups mentioned above are coregular, which means that there exist algebraically independent homogeneous $G$-invariant polynomials $\theta_1,\theta_2$ and $\theta_3$ such that every $G$-invariant polynomial can be represented in the form $Q(\theta_1,\theta_2,\theta_3)$ for some $Q\in\mathbb R[x_1,x_2,x_3]$. The polynomials $\theta_1,\theta_2$ and $\theta_3$ are called \emph{primary invariants}. 

Let us consider the canonical representations of groups $T_d, O_h$ and $I_h$, i.e., such that the vertices of the tetrahedron at $(1,1,1)$, $(1,-1,-1)$, $(-1,1,-1)$ and $(-1,-1,1)$ define the 3-fold axes for $T_d$, the vertices of the icosahedron at $(\pm\tau,\pm 1,0)$, $(0,\pm\tau,\pm 1)$ and $(\pm 1,0,\pm\tau)$ define the 5-fold axes for $I_h$, and the 4-fold axes for $O_h$ are given by the $x$, $y$ and $z$ axes. By $\tau$ we denote the \emph{golden ratio}, $\tau:=(1+\sqrt{5})/2$.  Let $I_{2}:=x^{2}+y^{2}+z^{2}$, $I_{3}:=xyz$,
$I_{4}:=x^{4}+y^{4}+z^{4}$, $I_{6}:=x^{6}+y^{6}+z^{6}$, $I_{6}^{\prime}
:=(\tau^{2}x^{2}-y^{2})(\tau^{2}y^{2}-z^{2})(\tau^{2}z^{2}-x^{2})$ and
$I_{10}:=(x+y+z)(x-y-z)(y-z-x) (z-y-x)(\tau^{-2}x^{2}-\tau^{2}y^{2})(\tau
^{-2}y^{2}-\tau^{2}z^{2})(\tau^{-2}z^{2}-\tau^{2}x^{2})$. Then $I_2,I_3$ and $I_4$ are primary invariants for $T_d$, $I_2, I_4$ and $I_6$ -- primary invariants for $O_h$, and $I_2, I_6'$ and $I_{10}$ -- primary invariants for $I_h$ \cite{Jaretal84}. 

We shall frequently use in the proof the facts gathered in Table \ref{priminv}:
\begin{table}[ht]\label{priminv}
	\begin{small}
		\centering
		\begin{tabular}
			[c]{c|c|c}
			primary invariant & MIN & MAX\\\hline \hline
			$I_{4}$, $I_6$ & cube & octahedron \\\hline
			$I_{6}^{\prime}$ & icosahedron & dodecahedron \\\hline
			$I_{10}$ &non-Archimedean vertex
			truncated icosahedron & dodecahedron 
		\end{tabular}
		
		\caption{Global minimizers and maximizers of the primary invariants for $O_h$ and $I_h$ restricted to the unit sphere}\label{priminv}
	\end{small}
\end{table}

\begin{proof}[Proof of Theorem \ref{HSThm}]
	We proceed in a similar way as before, applying the Hermite interpolation method to the function $h:[-1,1]\to\mathbb R^+$ defined by (\ref{entropyBloch}). The set of interpolating points is given by $T_v:=\{w\cdot u|u\in Gv\}$, where $v$ is the Bloch vector of one of the states constituting the POVM, $w$ is the Bloch vector indicated in the theorem's statement, $G=O_h$ for cube, octahedron and cuboctahedron and $G=I_h$ for icosahedron, dodecahedron and icosidodecahedron. The number of elements of the POVM can be now expressed as $k_v:=|G|/|G_v|$, where $G_v$ stands for the stabilizer of $v$ under the action of $G$. Note that $T_v\subset(-1,1)$; thus; after choosing the interpolating polynomial $p_v$ to agree with $h$ in every $t\in T_v$ up to the first derivative, we get  $\deg p_v<2|T_v|$ and $p_v$ interpolates $h$ from above, see Lemma \ref{fromabove} in Appendix \ref{hermite}. 
	The polynomial $P_v$ given by $P_v(u):=\ln(|G|/(2 |G_v|))+((2 |G_v|)/|G|)\sum_{[g]\in G/G_v}p_v(u\cdot gv)$ for $u\in B(2)$ is $G$-invariant and thus can be expressed in terms of primary invariants. Note that $P_v(u)\geq H_{\Pi}(u)$ for all $u\in B(2)$ and equality holds for all vectors from the orbit $Gv$. Thus, in order to see that the orbit of $w$ maximizes $H_{\Pi}$, it suffices to show that $P_v$ attains its maximum value in the orbit of $w$.
	
	Notice also that in Case \ref{platonic} we get the same polynomials $p_{v_1}$ and $p_{v_2}$  for dual POVMs (as the sets of interpolating points $T_{v_1}$ and $T_{v_2}$ coincide), but $P_{v_1}$ and $P_{v_2}$ differ.

	{\bf Cube} and {\bf octahedron.} This case is straightforward, as $|T|=2$ implies that the degree of the interpolating polynomials $p_{v_1}$ and $p_{v_2}$ is at most 3, and, as they are polynomial functions of $I_2$, $I_4$ and $I_6$, this degree is actually smaller than 2. Hence both $P_{v_1}$ and $P_{v_2}$ need to be constant on the sphere $B(2)$.
	
	{\bf Icosahedron} and {\bf dodecahedron.} In these cases we have $|T|=4$ and $\deg p_{v_1}=\deg p_{v_2}\leq 7$. In consequence, $P_{v_1}|_{^{S^{2}}}=A_1  +B_1I_6'$ and $P_{v_2}|_{^{S^{2}}}=A_2 +B_2 I_6'$ for some $A_1, A_2,B_1,B_2\in\mathbb R$. Thus, knowing that $I_6'$ restricted to the unit sphere attains its global maxima at the orbit of $v_2$, i.e., the vertices of dodecahedron, and global minima at the orbit of $v_1$, i.e., the vertices of icosahedron, see Table \ref{priminv},  it suffices to show that $B_1>0$ and $B_2<0$. We find their values  by calculating $P_v$ at two chosen points from different orbits and check that they are indeed of the desired sign.
	
	{\bf Cuboctahedron.} For the cuboctahedral measurement we get $|T|=3$ and $\deg P_v\leq\deg p_v\leq 5$. Thus, $P_v|_{^{S^{2}}}=A+B I_4$ for some $A,B\in\mathbb R$. Proceeding as in the previous case, it is enough to find the value of $B$ and check that it is positive.
	
	{\bf Icosidodecahedron.} In this case we get $|T|=5$ and so $\deg P_v\leq\deg p_v\leq 9$. In consequence, $P_{v}|_{^{S^{2}}}=A +BI_6'$ for some $A,B\in \mathbb R$. Again, we calculate $B$ using the same method as before and check that it is negative.
	
	To see the uniqueness of given global maximizers of the entropy, let us observe that if a pure state $\tilde{\rho}$ is also a global maximizer of the entropy, then $\widetilde{T_v}:=\{\tilde{w}\cdot u|u\in Gv\}\subset T_v$, where $\tilde{w}$ is a normalized Bloch vector corresponding to $\tilde{\rho}$. Using the fact that $\sum_{u\in Gv}\tilde{w}\cdot u=0$ and $\sum_{u\in Gv}(\tilde{w}\cdot u)^2=|Gv|/3$ (as $Gv$ is a 2-design) we get not only $\widetilde{T_v}= T_v$, but also we deduce that this equality need to hold even if we treat $\widetilde{T_v}$ and $T_v$ as multisets. It is easy to see now that there are no other vectors in $S^2$ with this property. 
\end{proof}

\begin{R}
	Note that the result for the cubical and octahedral measurements can be straightforwardly generalized to the case of the Havrda-Charv\'at-Tsallis $\alpha$-entropy or the R\'enyi $\alpha$-entropy with $\alpha\in(0,2)$ since the same argument as the one used in Remark \ref{genentr} holds here as well. 
\end{R}

\section{MAX, MIN and average relative entropy}
Let us recall that the relative entropy is strongly related to the  informational power of POVM, denoted by $W(\Pi)$ and defined as the maximum of mutual information between an ensemble of quantum states $\mathcal E$ and the POVM $\Pi$ taken over all possible ensembles \cite{DAretal11,Oreetal11}:
\begin{equation}\label{infpower}
W(\Pi)=\max_\mathcal E I(\mathcal E,\Pi)=\sum_{i=1}^m\eta(\sum_{j=1}^kP_{ij})+\sum_{j=1}^k\eta(\sum_{i=1}^m P_{ij})+\sum_{i=1}^m\sum_{j=1}^k \eta(P_{ij}),
\end{equation}
where $\mathcal E=\{\pi_i,\sigma_i\}_{i=1}^m$, $\sum_{i=1}^m \pi_i=1$, $\pi_i\geq 0$ for $i=1,\ldots,m$, $\sigma_i\in\mathcal S(\mathbb C^d)$ for $i=1,\ldots,m$ and $P_{ij}=\pi_i\tr(\sigma_i\Pi_j)$. 
Firstly, the maximum relative entropy  provides an upper bound on the informational power:
$$W(\Pi)\leq \max_{\rho\in\mathcal S(\mathbb C^d)}\widetilde{H}(\rho,\Pi),$$ which is saturated whenever there exists an ensemble $\mathcal E$ consisting of maximizers of the relative entropy that fulfils the condition $\tr(\sum_{i=1}^m(\pi_i\rho_i)\Pi_j)=1/k$ for $j=1,\ldots,k$, see \cite[p.\ 578]{SloSzy16}. Secondly, and more importantly in the context of this paper,  the lower bound for $W(\Pi)$ derived in \cite{DAr14} in the case of rank-1 POVMs is equal to the average relative entropy of $\Pi$ over all pure states, which can be easily calculated using Jones' formula \cite{Jon91,SloSzy16}: $$\langle \widetilde{H}(\rho,\Pi)\rangle _{\rho\in\mathcal{P}(
	\mathbb{C}^{d})  }    =\ln d-\sum_{j=2}^{d}\frac{1}{j}\leq W(\Pi).$$ This average tends to $1-\gamma$ when $d\rightarrow\infty$, where $\gamma\approx0.57722$ denotes the Euler-Mascheroni constant.

Note that  in eq.\ (\ref{infpower}) the maximum can be taken over ensembles consisting of pure states only and, what is more, if a POVM is group-covariant, then there exists also a group-covariant maximally informative ensemble \cite{Oreetal11}. Moreover, in such case the mutual information between a group-covariant ensemble and the POVM is equal to the relative entropy of POVM evaluated in any state from the ensemble.  Taking this all into account, we conclude that for any group-covariant SIC-POVM $\Pi=\{\Pi_j\}_{j=1}^{d^2}$ (note that all known SIC-POVMs are group-covariant) the so called `pretty-good' ensemble $\mathcal E=\{1/d^2,d\Pi_j\}_{j=1}^{d^2}$, already proven to be suboptimal \cite{DAretal14}, turns out to be `pretty-bad', being in fact the worst possible choice from the set of ensembles among which we are sure to find an optimal one, since it consists of the minimizers of the relative entropy, as  shown in Theorem \ref{maxsic}.   

The results derived in this paper combined with the previous ones concerning the maximizers of the relative entropy\cite{DAretal11,Oreetal11,SloSzy16,Szy14,SzySlo16} give us
a deeper insight into the behaviour of the relative entropy. In particular, we can now complete Table 5 from  \cite{SloSzy16}  to obtain  detailed information on the relative entropy of all highly symmetric POVMs in dimension 2, see Table \ref{hstable}. For better visualization see Fig.\ 3 in \cite{SloSzy16}. 

\begin{table}[ht]
	\begin{small}
		\centering
		\begin{tabular}
			[c]{c|c|c|c|c}
			POVM configuration  & minimal configuration & minimum & maximal configuration & maximum\\
			\hline\hline
			digon  & equator & $0$ & digon & $0.69315$\\
			regular $n$-gon ($n\rightarrow\infty$)  & digon & $0$ & dual $n$-gon & $0.30685$\\
			tetrahedron  & tetrahedron & $0.14384$ & `twin' tetrahedron & $0.28768$\\
			octahedron  & cube & $0.17744$ & octahedron & $0.23105$\\
			cube  & octahedron & $0.17744$& cube & $0.21576$\\
			cuboctahedron  & octahedron & $0.18443$ & cuboctahedron & $0.20273$\\
			icosahedron & dodecahedron & $0.18997$ & icosahedron & $0.20189$ \\
			dodecahedron  & icosahedron & $0.18997$ & dodecahedron & $0.19686$\\
			icosidodecahedron  & icosahedron & $0.19099$ & icosidodecahedron & $0.19486$\\\hline\hline
			average relative entropy & \multicolumn{4}{c}{$0.19315$}
		\end{tabular}
		
		
		\caption{The approximate values of informational power (maximum relative entropy) and minimum relative entropy on pure states (up to five digits)  for all types of HS-POVMs in dimension two together with both extremal configurations on the Bloch sphere.}\label{hstable}
	\end{small}
\end{table}

Now let us take a closer look at SIC-POVMs (Table \ref{sictable}). Although Theorem \ref{maxsic} gives us the minimum value of relative entropy and the description of minimal configurations for an arbitrary dimension $d$, there are only four cases for which the maximal configurations are known as well, and all these cases are in some sense exceptional among SIC-POVMs: 2-dimensional, Hesse's and Hoggar's SIC-POVMs are the only `supersymmetric' SIC-POVMs \cite{Zhu15} and the generic 3-dimensional SIC-POVMs form the only known infinite family of nonequivalent SIC-POVMs. Moreover, there is a numerical evidence \cite{DAr14} that in other cases the upper bound for relative entropy which can be derived from Theorem 2.5 in \cite{HarTop01}, see also Corollary 2 of \cite{DAr15}, is far from being tight. Interestingly, while we have to struggle with finding the maximum relative entropy for an arbitrary SIC-POVM and the description of maximal configurations does not show any visible pattern so far, the minimal configurations can be nicely described for any dimension. 

\begin{table}
	\begin{small}
		\centering
		\begin{tabular}
			[c]{c|c|c|c|c|c}
			\multirow{2}{*}{dimension}  & minimal  & \multirow{2}{*}{minimum} & \multirow{2}{*}{average} & maximal  & \multirow{2}{*}{maximum}\\
			& configuration & & & configuration  & \\
			\hline\hline
			2  & 2-SIC & 0.14384 & 0.19315 & `twin' 2-SIC & 0.28768\\
			3 (generic)   & generic 3-SIC & 0.17442 & 0.26528 & orthonormal basis & 0.40547\\
			3 (Hesse) & Hesse 3-SIC & 0.17442 & 0.26528 & complete 3-MUB & 0.40547 \\
			8 (Hoggar)  & Hoggar 8-SIC & 0.15687 & 0.36158 & `twin' Hoggar 8-SIC & 0.57536\\ 
			\hline
			$d$  &$d$-SIC & $\ln d-\frac{d-1}{d}\ln(d+1)$& $\ln d-\sum_{j=2}^d\frac{1}{j}$ & ? & $\leq\ln\frac{2d}{d+1}$\\ \hline
			\multirow{2}{*}{$d\to\infty$}  & & \multirow{2}{*}{$0$} & $1-\gamma$ & & $\leq\ln 2$ \\
			& & & $\approx 0.42278$ & & $\approx 0.69315$ 
		\end{tabular}
		
		
		\caption{The approximate values of informational power (maximum relative entropy), minimum relative entropy on pure states and its average (up to five digits) together with both extremal configurations for all fully solved cases of SIC-POVMs. The upper bound on relative entropy for an arbitrary $d$ given in the table is probably not achievable in general.}\label{sictable}
	\end{small}
\end{table}
\section{Acknowledgements}
The author is grateful to Wojciech S\l omczy\'nski and Anna Szczepanek for helpful discussions and valuable suggestions
for improving this paper. 

\appendix
\section{Hermite interpolation}\label{hermite}

Let us consider a sequence of pairs $\{(t_i,k_i)\}_{i=0}^m$, where $a\leq t_0<t_1<\ldots<t_m\leq b$ and $k_i$ are positive integers. Let $f:[a,b]\to\mathbb R$ be $D$ times differentiable, where $D:=k_0+k_1+\ldots+k_m$. There exists the only polynomial $p:[a,b]\to\mathbb R$ of degree less than $D$ agreeing with $f$ at points $t_i$ up to the $(k_i-1)$-th derivative, i.e., $p^{(k)}(t_i)=f^{(k)}(t_i)$ for $0\leq k<k_i$ and $i=0,1,\ldots,m$.
\begin{Lem}[Error in the Hermite interpolation]\label{HermiteError} For each $t\in (a,b)$ there exists $\xi\in (a,b)$ such that $\min\{t_0,t\}<\xi<\max\{t_m,t\}$ and
	$$f(t)-p(t)=\frac{f^{(D)}(\xi)}{D!}\prod_{i=0}^m(t-t_i)^{k_i}.$$
\end{Lem}
\begin{proof}
	See, e.g., \cite{StoBul}. 
\end{proof}
Let us now assume that all derivatives of $f$ of even order are strictly negative in $(a,b)$ and these of odd order greater than 1 are strictly positive, i.e.,
\begin{equation}\label{derivatives}
f^{(2l)}(x)<0,\quad f^{(2l+1)}(x)>0 \quad \textnormal{for } l=1,2,\ldots,\ x\in (a,b).
\end{equation}  It is easy to check that the entropy function $\eta$ fulfils these requirements.

\begin{Lem}\label{fromabove}
	Let $t_0>a$, $k_i=2$ for $i=0,\ldots, m-1$ and $k_m=1$ whenever $t_m=b$ or $k_m=2$ otherwise. Then under assumption (\ref{derivatives}) the Hermite polynomial $p$ interpolates $f$ from above, i.e., $p(t)\geq f(t)$ for $t\in[a,b]$. Moreover, $f(t)=p(t)$ if and only if $t=t_i$ for some $i=0,1,\ldots,m$.
\end{Lem}

\begin{proof}
	The proof follows straightforwardly from Lemma \ref{HermiteError} and (\ref{derivatives}). 
\end{proof}
Note that this lemma is an analogue of the requirements given in \cite{SloSzy16} to obtain an interpolation from below.


\begin{thebibliography}{99}



\bibitem{Benetal15} Bent, N., Qassim, H., Tahir, A.A., Sych, D., Leuchs, G., S\'anchez-Soto, L.L.,  Karimi, E., Boyd, R.W., ``Experimental realization of quantum tomography of photonic qudits via symmetric
  informationally complete Positive Operator-Valued Measures", Phys.\ Rev.\ X 5, 041006 (2015).


\bibitem{Cavetal02} Caves, C.M., Fuchs, C.A., Schack, R., ``Unknown quantum states: The quantum de Finetti representation'', J. Math. Phys. 43, 4537-4559 (2002).

\bibitem{Cheetal15} Chen, B., Li, T., Fei, S.-M., ``General SIC measurement-based entanglement detection", Quantum Inf.\ Process.\ 14, 2281-2290 (2015).

\bibitem{DAr14} Dall'Arno, M., ``Accessible information and informational power of quantum 2-designs", Phys.\ Rev.\ A 90, 052311 (2014).

\bibitem{DAr15} Dall'Arno, M., ``Hierarchy of bounds on accessible information and informational power", Phys.\ Rev.\ A 92, 012328 (2015).

\bibitem{DAretal11} Dall'Arno, M., D'Ariano, G.M., Sacchi, M.F., ``Informational power of quantum measurement", Phys.\ Rev.\ A 83, 062304 (2011).

\bibitem{DAretal14} Dall'Arno, M., Buscemi, F., Ozawa, M., ``Tight bounds on the accessible information and informational power", J.\ Phys.\ A 47, 235302 (2014).

\bibitem{DerKem02}Derksen, H., Kemper, G., Computational Invariant
Theory. Encyclopaedia of Mathematical Sciences 130, Springer-Verlag,
Berlin (2002).




\bibitem{Gra04} Grassl, M.,`` On SIC-POVMs and MUBs in dimension 6", arXiv:quant-ph/0406175

\bibitem{HarTop01} Harremo\"es, P., Tops\o e, F., ``Inequalities between entropy and index of coincidence derived from information diagrams", IEEE Trans.\ Inf.\ Theory 47, 2944-2960
(2001).

\bibitem{HeiZim} Heinosaari, T., Ziman, M., ``The Mathematical Language
of Quantum Theory: From Uncertainty to Entanglement", Cambridge UP, Cambridge (2011).

\bibitem{Jaretal84}Jari\'{c}, M.V., Michel, L., Sharp, R.T., ``Zeros of covariant
vector fields for the point groups: invariant formulation", J.\ Phys.\ (Paris) 45, 1--27 (1984).

\bibitem{JiaPhd} Jiangwei, S., Ph.D. thesis, National University of Singapore, Singapore, 2013.

\bibitem{Jon91}Jones, K.R.W., ``Quantum limits to information about states for
finite dimensional Hilbert space", J.\ Phys.\ A \textbf{24}, 121--130 (1991).



\bibitem{Oreetal11} Oreshkov, O., Calsamiglia, J., Mu\~{n}oz-Tapia, R., Bagan, E., ``Optimal signal states for quantum detectors", New J.\ Phys.\ 13, 073032 (2011).

\bibitem{Pucetal15} Pucha\l a, Z., Rudnicki, \L., Chabuda, K., Paraniak, M., \.Zyczkowski, K., ``Certainty relations, mutual entanglement and non-displacable manifolds", Phys.\ Rev.\ A 92, 032109 (2015).

\bibitem{Ren05} Renes, J.M., ``Equiangular spherical codes in quantum cryptography", Quant.\ Inf.\ Comput.\ 5, 80–91 (2005).

\bibitem{Renetal04} Renes, J.M., Blume-Kohout, R., Scott, A.J., Caves, C.M., ``Symmetric informationally complete quantum measurements", J.\ Math.\ Phys.\ 45, 2171-2180 (2004).

\bibitem{San95} S\'anchez-Ruiz, J., ``Improved bounds in the entropic uncertainty and certainty relations for complementary observables", Phys.\ Lett.\ A 201, 125-131 (1995). 

\bibitem{Sco06} Scott, A.J., ``Tight informationally complete quantum measurements", J.\ Phys.\ A 39, 13507–13530 (2006).

\bibitem{SloSzy16} S\l omczy\'nski, W., Szymusiak, A., ``Highly symmetric POVMs and their informational power", Quantum Inf.\ Process.\ 15, 565-606 (2016).

\bibitem{StoBul} Stoer, J., Bulirsch, R., ``Introduction to numerical analysis", Springer, New York (2002).

\bibitem{Szy14} Szymusiak, A., ``Maximally informative ensembles for SIC-POVMs in dimension 3", J.\ Phys.\ A 47, 445301 (2014).

\bibitem{SzySlo16} Szymusiak, A., S\l omczy\'nski, W., ``Informational power of the Hoggar symmetric informationally complete positive operator-valued measure", Phys.\ Rev.\ A 94, 012122 (2016).


\bibitem{Zhu15} Zhu, H., ``Super-symmetric informationally complete measurements", Ann.\ Phys.\ (NY) 362, 311-326 (2015).

\bibitem{ZhuEng11} Zhu, H., Englert, B.-G., ``Quantum state tomography with fully symmetric measurements and
product measurements", Phys.\ Rev.\ A 84, 022327 (2011).

\end{thebibliography}
\end{document}